\documentclass[letterpaper,10pt]{article}
\usepackage[T1]{fontenc}
\usepackage{parskip}
\usepackage[font={small},labelfont=bf]{caption}

\usepackage[margin=2.5cm]{geometry}
\usepackage{times}
\usepackage{amsmath,amssymb,amsthm}
\usepackage[affil-it]{authblk}
\usepackage{listings}

\usepackage{xspace}
\usepackage{graphicx}
\usepackage{color}

\usepackage{algpseudocode}

\usepackage{pgfplots}
\usepgfplotslibrary{statistics}
\pgfplotsset{compat=newest}
\newlength\figureheight
\newlength\figurewidth
\def\Tplt/{20}

\def\Tobs/{30}

\usepackage{glossaries}
\newacronym{pem}{PEM}{prediction error methods}
\newacronym{oe}{OE}{output-error}
\newacronym{pdf}{pdf}{probability density function}
\newacronym{gp}{GP}{Gaussian process}
\newacronym{dc}{DC}{diagonal correlated}
\newacronym{tc}{TC}{tuned-correlated}
\newacronym{ss}{SS}{stable spline}
\newacronym{brm}{BRM}{Bayesian risk minimization}

\usepackage[hidelinks]{hyperref}
\usepackage{algorithm}
\usepackage[capitalize, noabbrev]{cleveref}
\crefformat{equation}{(#2#1#3)}
\crefrangeformat{equation}{(#3#1#4) to~(#5#2#6)}
\crefmultiformat{equation}{(#2#1#3)}%
{ and~(#2#1#3)}{, (#2#1#3)}{ and~(#2#1#3)}

\usetikzlibrary{external}
\makeatletter \newcommand{\pgfplotsdrawaxis}{\pgfplots@draw@axis} \makeatother
\pgfplotsset{axis line on top/.style={
  axis line style=transparent,
  ticklabel style=transparent,
  tick style=transparent,
  axis on top=false,
  after end axis/.append code={
    \pgfplotsset{axis line style=opaque,
      ticklabel style=opaque,
      tick style=opaque,
      grid=none}
    \pgfplotsdrawaxis}
  }
}

\newif\ifarxiv

\graphicspath{{fig/}}

\newcounter{exmpcnt}
\newtheorem{exmp}[exmpcnt]{Example}

\newcounter{resultcnt}
\newtheorem{thm}[resultcnt]{Result}

\newcommand{\ie}{i.e.\xspace}

\DeclareMathOperator{\E}{E}
\DeclareMathOperator{\Cov}{Cov}

\DeclareMathOperator{\tr}{tr}
\DeclareMathOperator*{\argmin}{arg\,min}
\DeclareMathOperator*{\argmax}{arg\,max}

\newcommand{\trace}[1]{\tr\left(#1\right)}
\newcommand{\diag}[1]{\operatorname{diag}\left(#1\right)}

\newcommand{\Exp}[1]{\E\left[#1\right]}
\newcommand{\CExp}[2]{\E_{#1}\left[#2\right]}
\newcommand{\mbf}[1]{\mathbf{#1}}
\newcommand{\mbs}[1]{\boldsymbol{#1}}
\newcommand{\what}[1]{\widehat{#1}}
\newcommand{\wtilde}[1]{\widetilde{#1}}

\newcommand{\wbar}[1]{\overline{#1}}
\renewcommand{\vec}[1]{\mbs{#1}} 
\newcommand{\mat}[1]{\mbf{#1}} 
\newcommand{\I}{\mat{I}}
\newcommand{\0}{\mat{0}}
\newcommand{\bbm}{\begin{bmatrix}}
\newcommand{\ebm}{\end{bmatrix}}

\newcommand{\Npdf}[3]{\mathcal{N}\left(#1 \middle| #2,#3\right)}

\newcommand{\T}{\top}
\newcommand{\risk}{\mathcal{R}}
\newcommand{\loss}{\mathcal{L}}
\newcommand{\pinv}{-}

\newcommand{\impsc}{g}
\newcommand{\imp}{\vec{\impsc}}

\newcommand{\impmean}{\wbar{\imp}}
\newcommand{\impcov}{\mat{\Sigma}}

\newcommand{\inputsc}{u}

\newcommand{\datasc}{y}
\newcommand{\data}{\vec{\datasc}}
\newcommand{\modeldata}{\wtilde{\vec{\datasc}}}

\newcommand{\dataset}{\mathcal{D}}
\renewcommand{\H}{\mat{H}}
\newcommand{\kernelK}{\mat{K}}

\newcommand{\noisesc}{e}
\newcommand{\noise}{\vec{\noisesc}}
\newcommand{\noisevar}{\lambda}

\newcommand{\para}{\vec{\theta}}
\newcommand{\paradec}{\what{\para}}
\newcommand{\hyp}{\vec{\eta}}
\newcommand{\modellinear}{\mathcal{G}}
\newcommand{\weight}{\mat{W}}

\newcommand{\finalT}{N}

\arxivtrue

\title{Regularized parametric system identification: \\ a decision-theoretic formulation}
\author[1]{Johan W{\aa}gberg \thanks{\url{johan.wagberg@it.uu.se}}}
\author[1]{Dave Zachariah\thanks{\url{dave.zachariah@it.uu.se}}}
\author[1]{Thomas B. Sch\"on\thanks{\url{thomas.schon@it.uu.se}}}
\affil[1]{Department of Information Technology, Uppsala University}

\begin{document}

\maketitle

%
%
%

\begin{abstract}
\noindent Parametric prediction error methods constitute a classical approach to the identification of linear dynamic systems with excellent large-sample properties.
A more recent regularized approach, inspired by machine learning and Bayesian methods, has also gained attention.
Methods based on this approach estimate the system impulse response with excellent small-sample properties.
In several applications, however, it is desirable to obtain a compact representation of the system in the form of a parametric model.
By viewing the identification of such models as a decision, we develop a decision-theoretic formulation of the parametric system identification problem that bridges the gap between the classical and regularized approaches above.
Using the output-error model class as an illustration, we show that this decision-theoretic approach leads to a regularized method that is robust to small sample-sizes as well as overparameterization.
\end{abstract}



\section{Introduction}
Traditionally, methods that identify a parametric model of a dynamical system are formulated using the frequentist paradigm.
A classical parametric approach is the \gls{pem}, which has excellent large-sample properties \cite{Soderstrom&Stoica1988_system,Ljung1998_sysid}.
Computational advances have more recently enabled identification methods formulated in the Bayesian paradigm \cite{Peterka1981bayesian,Ninness&Henriksen2010_bayesian}.
A recent nonparametric approach is the kernel regression method, which estimates system impulse responses in a regularized manner with excellent small-sample properties \cite{Pillonetto&Giuseppe2010_kernel,PillonettoEtAl2014_kernelsurvey}.
The choice of regularization corresponds to the modelling of prior information about the system under consideration \cite{Rasmussen&Williams2006_gaussian,Murphy2012machine}.

The classical and the Bayesian approach differ radically in their underlying philosophies, practical interpretations, and target models.
In this paper we propose a way to bridge the gap by developing a decision-theoretic approach for parametric system identification, cf. \cite{Chernoff&Moses1986_elementary,Berger1985_statistical,Robert2007_bayesian}.
Specifically, by viewing the identification of a specific parametric system model as a decision, we can unify classical frequentist and modern Bayesian identification approaches.
The framework enables comparisons and fruitful cross-fertilization of ideas from both approaches.
Specifically, it dispels
\begin{itemize}
	\item the need to specify a prior distribution of the model parameters, and
	\item the assumption that the system belongs to the considered model class.
\end{itemize}
Moreover, it achieves a natural regularization of the identification problem, which mitigates overfitting
\begin{itemize}
	\item in the small-sample regime
	\item when using overparameterized models.
\end{itemize}

\emph{Notation:} $\| \imp \|_{\weight} = \sqrt{\imp^\T \weight \imp }$ is a weighted $\ell_2$-norm.
The operation of extracting the diagonal elements of a square matrix as a vector is denoted $\diag{\cdot}$.

\section{Model classes}
\subsection{Nonparametric model specification}
Consider a discrete linear time-invariant and causal dynamic system with a single input and a single output (SISO).
The relationship between the input $\inputsc(t)$ and the output $\datasc(t)$ of such a system can be written as
\begin{equation}
	\datasc(t) = G_{\star}(q) \inputsc(t) + \noisesc(t),\quad t = 1, \ldots, \finalT,
	\label{eq:LTI}
\end{equation}
where $G_{\star}(q)$ is the system transfer function, $q$ is the shift operator and $\noisesc(t)$ is a random zero-mean disturbance assumed to be white with variance $\noisevar$ \cite{Ljung1998_sysid}.
We assume that the input $\inputsc(t)$ is known for all $t$ and that $\inputsc(t) \equiv 0$ for $t \leq 0$.
Let $$\dataset = \bigl\{ (\inputsc(1), \datasc(1)), \: \dots, \: (\inputsc(N), \datasc(N))\bigr\}$$ be input-output data collected at time instants $t = 1, \ldots, \finalT$.

The transfer function can be written as
\begin{equation}
    G_{\star}(q) = \sum_{k=0}^{\infty} \impsc_{\star}(k)q^{-k}.
    \label{eq:transfer_function_expansion}
\end{equation}
The collection $\impsc_{\star} = \{\impsc_{\star}(k)\}^{\infty}_{k=0}$ is the impulse response of the system and a linear system is uniquely defined by its impulse response, which constitutes a nonparametric model of the system.
Let $\modellinear$ denote the class of all causal SISO systems.
Thus for \cref{eq:LTI} we have that $\impsc_{\star} \in \modellinear$, see \cref{fig:modelclass} for an illustration.

Since the input $\inputsc(t)$ is zero for $t \leq 0$, for finite time $\finalT$, \cref{eq:LTI} can be written as a finite sum
\begin{equation}
    \datasc(t) = \sum_{k=0}^{t-1} \impsc_{\star}(k) \inputsc(t-k) + \noisesc(t) \; \Leftrightarrow \; \data = \H \imp_{\star} + \noise,
    \label{eq:output_as_sum}
\end{equation}
where 
\begin{equation*}
    \data = \bbm \datasc(1) \\ \datasc(2) \\ \vdots \\ \datasc(\finalT) \ebm \; \text{and} \; \H = \bbm \inputsc(1) \\ \inputsc(2) & \inputsc(1) \\ \vdots & \vdots & \ddots  \\ \inputsc(\finalT) & \inputsc(\finalT - 1) & \cdots & \inputsc(1) \ebm,
\end{equation*}
and $\H_{i,j} = 0$ for $j > i$, contain $\dataset$. The $N$ first coefficients of the impulse response are
\begin{equation*}
  \imp_{\star} = \bbm \impsc_{\star}(0) \\ \impsc_{\star}(1) \\ \vdots \\ \impsc_{\star}(\finalT-1) \ebm
\end{equation*}
and $\noise = [\noisesc(1) \,\, \noisesc(2) \,\, \dots \,\, \noisesc(N)]^{\T}$ contains the disturbance samples, with mean value $\Exp{\noise} = \0$ and covariance $\Exp{\noise \noise^{\T}} = \noisevar \I$.

\subsection{Parametric model specification}
Estimating the transfer function or, equivalently, the impulse response $g_{\star}$ can give useful information about the properties of the system. 
However, in many applications, e.g. control, a more compact representation of a system is more useful.
We consider a family of impulse responses parameterized by a vector $\para$ and denote it as $\modellinear_{\para}$.
One possible family is to let the transfer function be a rational function in $q^{-1}$, \ie the ratio of two polynomials, 
\begin{equation}
    \begin{aligned}
        G_{\para}(q)
            = \frac{B_{\para}(q)}{F_{\para}(q)}
            &= \frac{b_0 q^{-n_k} + b_1 q^{-n_k - 1} + \cdots + b_{n_b}q^{-n_k - n_b}}{1 + f_1 q^{-1} + \cdots + f_{n_f}q^{-nf}},
    \end{aligned}
    \label{eq:modelparametric}
\end{equation}
and $\para = \bbm b_{0}, b_{1}, \ldots, b_{n_b}, f_{1}, f_{2}, \ldots, f_{n_f} \ebm$.
The orders of the two polynomials are determined by $n_b$ and $n_f$, respectively, and an initial time delay is controlled by $n_k$. Here we consider them fixed a priori, but they can be determined as well. The transfer function \cref{eq:modelparametric} has a corresponding impulse response denoted $\impsc_{\para}$.
The relation between the parametric model class $\modellinear_{\para}$ and the family of all causal SISO systems $\modellinear$ is illustrated in \cref{fig:modelclass}.
\begin{figure}
    \centering
	\ifarxiv
		\includegraphics{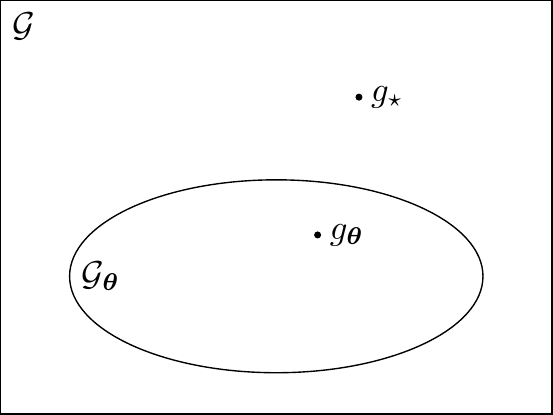}
	\else
		\tikzsetnextfilename{modelclass}
		\input{fig/modelclass.pgf}
	\fi

    \caption{The nonparametric class $\modellinear$ contains all linear time-invariant causal SISO systems. The parametrized class $\modellinear_{\para}$ is a subset of $\modellinear$. Note that $\impsc_\star$ does not have to belong to $\modellinear_{\para}$.}
    \label{fig:modelclass}
\end{figure}

\begin{exmp}[Two members in the  model classes]\label{ex:example} To illustrate the members of the model classes, suppose $\impsc_{\star} \in \modellinear$ is a second order system with no zeros, two poles in $z = 0.64 \pm 0.48i$ and a static gain of 2. Each member of a given class $\modellinear_{\para}$ is specified by the parameters
\begin{align*}
    \para 
        &= \bbm b_0 & f_1 & f_2 & f_3 & f_4\ebm^{\T}
\end{align*}
Note that this model class is overparameterized with respect to the unknown system. A specific example of $\impsc_{\para} \in \modellinear_{\para}$ is illustrated in \cref{fig:example_1_generating_system} along with $\impsc_\star$.

\begin{figure}
    \centering
    \setlength{\figurewidth}{0.55\columnwidth}
    \setlength{\figureheight}{0.3\columnwidth}
	\ifarxiv
		\includegraphics{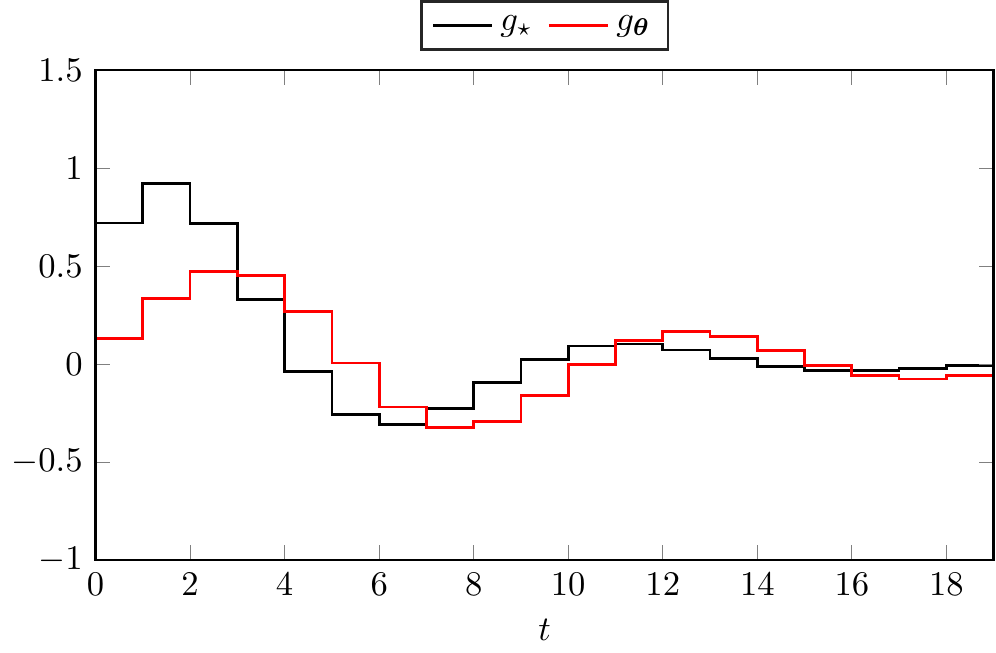}
	\else
		\tikzsetnextfilename{example_1_generating_system}
		\input{fig/example_1_generating_system.pgf}
	\fi

    \caption{The first \Tplt/ coefficients of impulse responses. For sake of illustration we chose $\para = \protect\bbm 0.13 & -2.65 & 2.92 & -1.64 & 0.41 \protect\ebm^{\T}$.} \label{fig:example_1_generating_system}
\end{figure}
\end{exmp}

\section{Decision-theoretic formulation}
Given the dataset $\dataset$, we want to identify a parametric model in $\modellinear_{\para}$ that is as close to the unknown system as possible. The choice of a specific model, specified by $\para$, is viewed as a decision with an associated loss
\begin{equation}
  \loss(\para) =  \frac{1}{2} \big\| \imp_{\star} - \imp_{\para} \big\|^{2}_{\weight},
\label{eq:loss}
\end{equation}
where $\imp_{\star}$ is given in \cref{eq:output_as_sum}, $\imp_{\para}$ is the corresponding impulse response of the parametric model and $\weight \succeq \0$ is a weight matrix.
In real applications, however, the system impulse response is unknown.
In lieu of $\imp_{\star}$, we consider an impulse response $\imp$ as a random variable drawn from $\modellinear$ and average the loss \cref{eq:loss} of a decision $\para$ over all possible values of $\imp$, a.k.a. risk:
\begin{equation}
	\risk(\para) = \CExp{\imp \mid \dataset}{ \loss(\para) }
	= \frac{1}{2}\trace{\weight \impcov} + \frac{1}{2}\|\impmean - \imp_{\para} \|^{2}_{\weight},
	\label{eq:risk}
\end{equation}
where $\impmean$ and $\impcov$ are the mean and covariance matrix of $\imp$, respectively, given $\dataset$.
The proof of the equality is given in \cref{app:risk}.
The expected value $\impmean$ represents our best guess of $\imp_\star$ and therefore it is reasonable to choose $\weight$ based on the precision of this guess.

The optimal decision is the parameter $\what{\para}$ that minimizes the risk, i.e.,
\begin{equation}
\boxed{\what{\para} = \argmin_{\para} \; \risk(\para).}
\label{eq:optimaldecision}
\end{equation}
This decision rule generates different identification methods based on model choice for $\imp$, which determines the mean $\impmean$ and precision matrix $\weight$. The problem \cref{eq:optimaldecision} can be solved using numerical search algorithms, cf. \cite[ch.~7]{Soderstrom&Stoica1988_system}. Note that the loss \cref{eq:loss} and resulting decision-rule can equivalently be formulated in the frequency domain. 

\begin{thm}
    Suppose we let $\impmean$ be the unbiased least squares estimate of $\imp_{\star}$ and take $\weight$ to be its precision matrix, i.e.,
    \begin{equation}
        \impmean = \H^{\dag}\data \quad\text{and}\quad  \weight = \noisevar^{-1} \H^{\T}\H.
        \label{eq:pem_mu_and_weight}
    \end{equation}
    Then \cref{eq:optimaldecision} corresponds to \gls{pem}. Note that the variance $\noisevar$ does not affect the decision and can therefore be set to unity. 
\end{thm}

\begin{proof}
See \cref{app:pem_as_decision}.
\end{proof}

\begin{exmp}[Making a decision]\label{ex:2}
To identify the system, we generate in input $u(t)$ as a Gaussian white process with unit variance.
The unknown system $\impsc_\star$ is the same as in the previous example and the unknown disturbance variance is $\lambda = 1$.
We also consider the same overparameterized model class $\modellinear_{\para}$.

Using the resulting input-output data $\dataset$, it is straightforward to compute $\impmean$ and $\weight$ in \cref{eq:pem_mu_and_weight}.
Here we used a gradient search algorithm to find $\what{\para}$ in \cref{eq:optimaldecision}.
The impulse response of the generating system $\imp_\star$ and the mean $\impmean$ are plotted in \cref{fig:example_1_pem} along with the risk-minimizing decision $\imp_{\what{\para}}$ given by \cref{eq:optimaldecision}.
The weighting $\weight$ is visualized here as uncertainty bands around $\impmean$.
Specifically, when the the matrix is full rank, we plot a band $\pm 2 \sqrt{\diag{\weight^{-1}}}$, where the square root is computed element wise.
The band then corresponds a dispersion of two standard deviations.
Narrower uncertainty bands correspond to higher weights and therefore forces the optimal decision $\imp_{\what{\para}}$ to be close to $\impmean$ at the corresponding coefficients.
\begin{figure}
    \centering
    \setlength{\figurewidth}{0.55\columnwidth}
    \setlength{\figureheight}{0.35\columnwidth}
	\ifarxiv
		\includegraphics{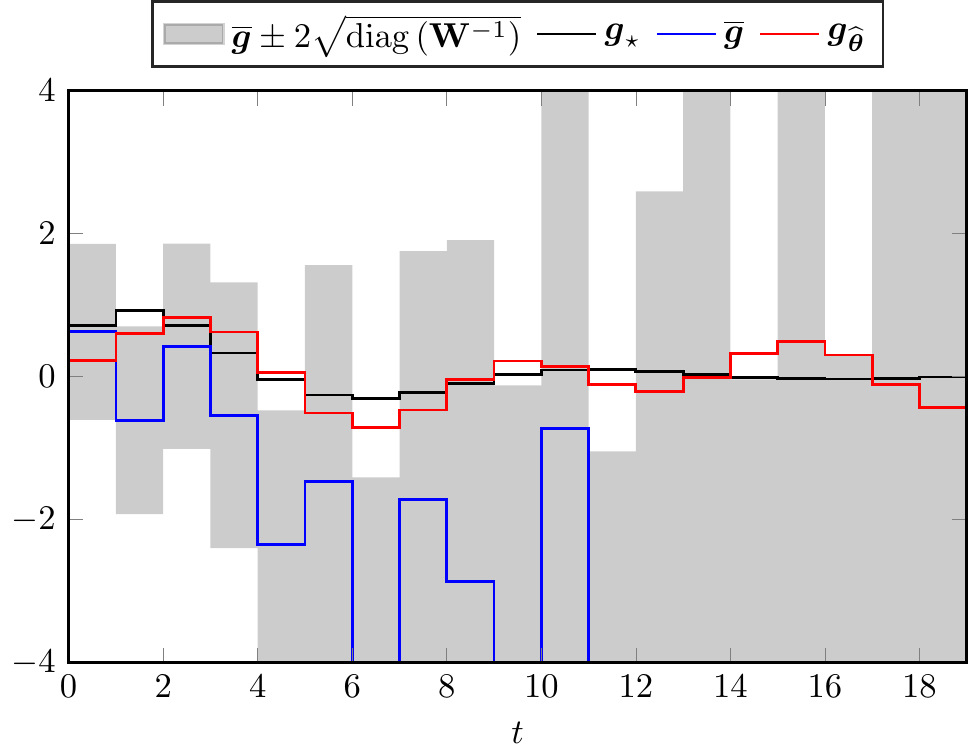}
	\else
		\tikzsetnextfilename{example_1_pem}
		\input{fig/example_1_pem.pgf}
	\fi

    \caption{Visualization of $\impmean$ (blue), the weighting $\weight$ (gray) and the final decision $g_{\paradec}$ (red) together with the data-generating impulse response $\imp_{\star}$ (black). Here $\impmean$ falls outside of the range of the figure.}
    \label{fig:example_1_pem}
\end{figure}
\end{exmp}

\begin{thm}
Suppose $\imp$ and the disturbances are modeled as Gaussian distributions, such that $\imp \sim \mathcal{N}(\0, \kernelK)$. Then 
    \begin{equation}
            \impmean = (\noisevar^{-1}\H^{\T}\H + \kernelK^{-1})^{-1}\H^{\T}\noisevar^{-1}\data \quad \text{and} \quad
            \weight = \kernelK^{-1} + \noisevar^{-1}\H^{\T}\H,
        \label{eq:prior_mu_and_weight}
    \end{equation}
    corresponds to the precision matrix of $\imp$ when conditioned on $\dataset$. Note that $\noisevar$ need not be set directly but can be absorbed into the parametrization of $\kernelK$.
\end{thm}

\begin{proof}
Let $\Npdf{\vec{x}}{\vec{\mu}}{\mat{\Sigma}}$ denote the \gls{pdf} of a multivariate Gaussian distribution with expected value $\vec{\mu}$ and covariance matrix $\mat{\Sigma}$.
Since $p(\data \mid \imp) = \Npdf{\data}{\H\imp}{\noisevar\I}$ and $p(\imp) = \Npdf{\imp}{\0}{\kernelK}$, the distribution over $\imp$ is conjugate to the data distribution and the result follows in a straightforward manner \cite{Rasmussen&Williams2006_gaussian}.
\end{proof}

\def \mytmp {\value{exmpcnt}}
\setcounter{exmpcnt}{1}
\begin{exmp}{(cont'd)}
    A priori we may expect the unknown system in $\modellinear$ to be stable. To model this prior information, suppose we set the entries of the covariance matrix  $\kernelK$ as 
    \begin{equation}
		\kernelK_{i,j} = 100 \cdot 0.8^{(i+j)/2} \cdot 0.7^{|i-j|}
		\label{eq:DC1}
    \end{equation}
In the subsequent section, we discuss means to automatically select $\kernelK$. The specific choice above places high prior weight on systems in $\modellinear$ that exhibit exponentially decaying impulse responses and weak positive correlation between adjacent coefficients. \Cref{fig:example_1_prior} shows the result of this choice. Compared to \Cref{fig:example_1_pem}, the uncertainty bands are now tighter around the mean $\impmean$, which also follows $\imp_\star$ more closely. Consequently, the risk-minimizing decision $\imp_{\what{\para}}$ is a better approximation of $\imp_\star$. The underlying reason is the regularization achieved by $\impmean$ and $\weight$. 
\begin{figure}
    \centering
    \setlength{\figurewidth}{0.55\columnwidth}
    \setlength{\figureheight}{0.35\columnwidth}
	\ifarxiv
		\includegraphics{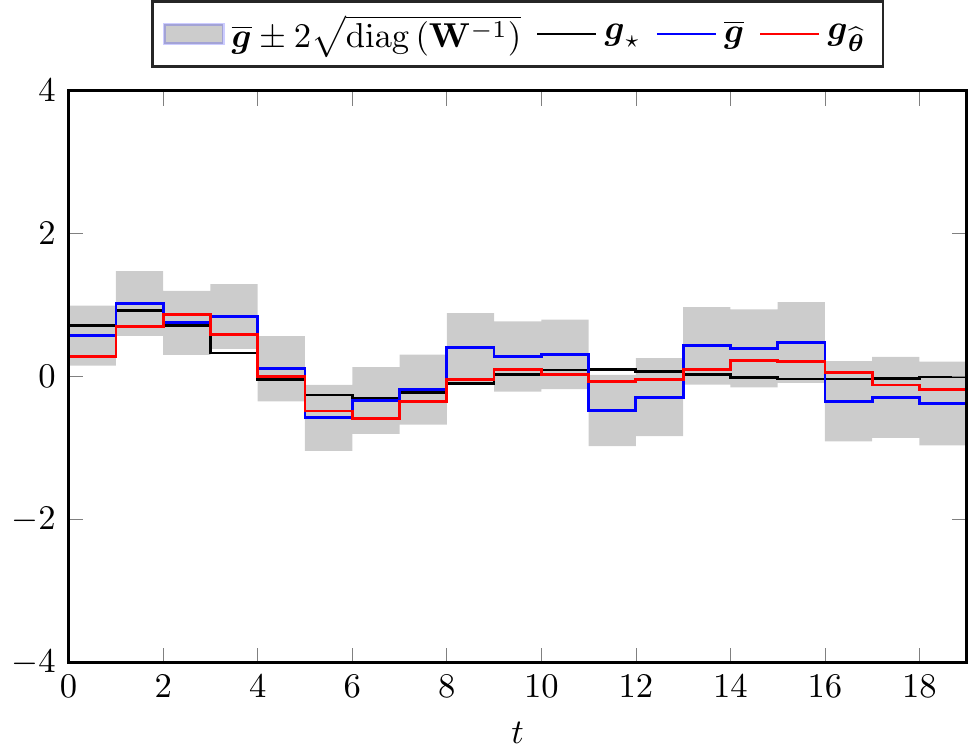}
	\else
		\tikzsetnextfilename{example_1_prior}
		\input{fig/example_1_prior.pgf}
	\fi

    \caption{Visualization of $\impmean$ (blue), the weighting matrix~$\weight$ (gray) and the final decision $g_{\hat\theta}$ (red) together with the data-generating impulse response $\imp_{\star}$ (black).}
    \label{fig:example_1_prior}
\end{figure}
\end{exmp}
\setcounter{exmpcnt}{\mytmp}
The example illustrates a bridge between classical frequentist and modern Bayesian approaches to system identification using a decision-theoretic approach outlined above. Note also that unlike a standard Bayesian approach, there is no need to directly specify a prior distribution over the parameters $\para$ nor does the data generating $\imp_{\star}$ have to reside in the considered model class $\modellinear_{\para}$. Via the covariance matrix $\kernelK$, the framework incorporates our prior information that the system is probably stable. This, however, begs the question of how to model such information by  $\kernelK$? This question is discussed in the subsequent section.

\section{Incorporating prior information}\label{sec:modeling_prior_information}
Prior information about the system is encoded in the prior distribution over the impulse response, e.g.,
\begin{equation*}
    \imp \sim \mathcal{N}\big( \0, \kernelK(\hyp) \big),
\end{equation*}
where the hyperparameter $\hyp$ specifies the covariance matrix. Let $\kernelK(\hyp)_{i,j}$ denote the $ij$th element of the matrix, which corresponds to a covariance function or `kernel' of the impulse response coefficients. A popular choice is
\begin{equation}
    \kernelK(\hyp)_{i,j} = c \alpha^{(i + j)/2} \rho^{|i - j|}, \quad
    \hyp = \bbm c & \alpha & \rho \ebm^{\T},
    \label{eq:dckernel}
\end{equation}
which represents exponentially decaying or increasing impulse responses and is thus suitable for modelling e.g. a priori information about stable systems. For other model choices, see \cite{PillonettoEtAl2014_kernelsurvey}.

The specification of hyperparameters $\hyp$ can be achieved in different ways. The Bayesian approach is to set them a priori, reflecting personalized beliefs about the system prior to seeing any data. A more pragmatic approach is to tune $\hyp$ using the data $\dataset$. Such methods include cross-validation, the SURE-estimator, and maximum likelihood \cite{Berger1985_statistical,Rasmussen&Williams2006_gaussian,Murphy2012machine,PillonettoEtAl2014_kernelsurvey}. Here we consider the maximum likelihood approach 
\begin{equation}
    \what{\hyp} = \argmax_{\hyp} \; \ell(\hyp),
\end{equation}
where $\ell(\hyp)$ is obtained by marginalizing out $\imp$ from the nonparametric data model. For Gaussian processes, we have that
\begin{equation}
    \ell(\hyp) =  -\frac{1}{2} \ln| \noisevar \I + \H\kernelK\H^{\T}| - \frac{1}{2}\|\data\|^{2}_{(\noisevar \I + \H \kernelK \H^{\T})^{-1}}
    \label{eq:marginal_likelihood}
\end{equation}
which only depends on $\kernelK(\hyp)_{i,j}$ evaluated at the observed time instants \cite{Rasmussen&Williams2006_gaussian}.

\begin{exmp}[Setting hyperparameters]
Consider again the dataset $\dataset$.
In~\cref{eq:DC1} we used the so-called \gls{dc} kernel with hyperparameter $\hyp_0 = \bbm 100 & 0.8 & 0.7\ebm^{\T}$.
Using this as initial guess and maximizing the marginal likelihood we obtain $\what{\hyp} = \bbm 0.45 & 0.71 & 0.74 \ebm^{\T}$.
We can visualize the prior by studying the marginal variance of each coefficient in the impulse response and plotting this as an uncertainty band together with the impulse response of the generating system $\imp_{\star}$.
While both priors encompass $\imp_{\star}$, the prior with tuned hyperparameters is much tighter as seen in \cref{fig:example_1_eb}.
The resulting resulting decision is illustrated in \cref{fig:example_1_krm}, which is to be compared with \cref{fig:example_1_prior}.
We see that the tuned hyperparameters lead to a more accurate $\impmean$ and weights $\weight$, which consequently improves the final decision.
Thus we can avoid the manual tuning of hyperparameters after choosing the form of $\kernelK(\hyp)$ that model the prior information about the system.
\begin{figure}
    \centering
    \setlength{\figurewidth}{0.55\columnwidth}
    \setlength{\figureheight}{0.35\columnwidth}
	\ifarxiv
		\includegraphics{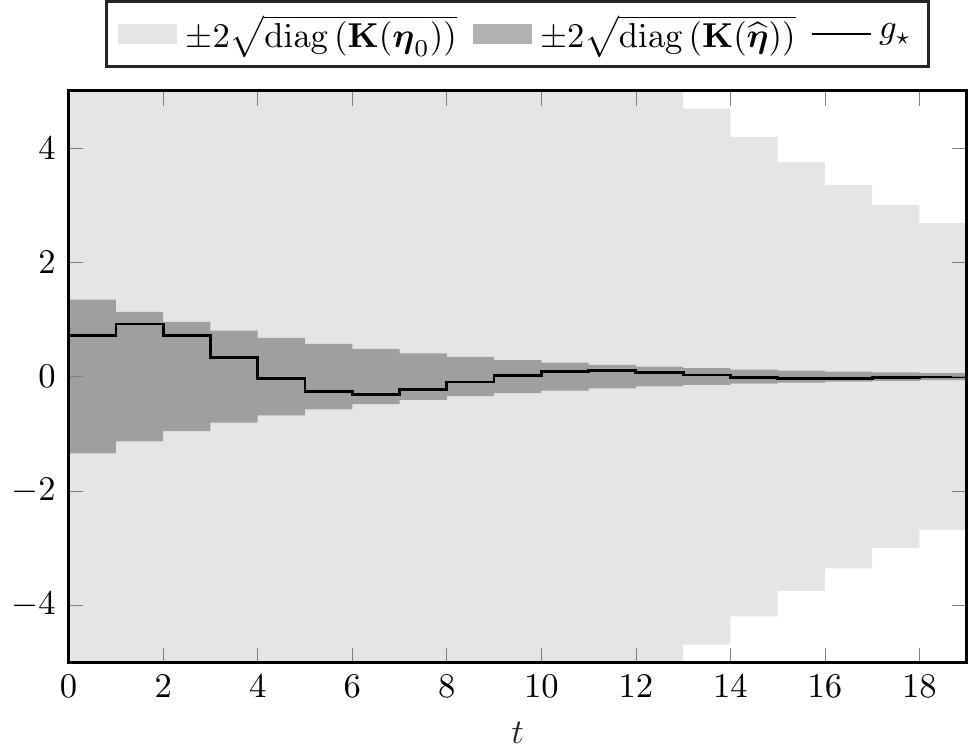}
	\else
		\tikzsetnextfilename{example_1_eb}
		\input{fig/example_1_eb.pgf}
	\fi

    \caption{A comparison of the prior distribution of $\imp$ using the arbitrarily chosen $\hyp$ in the example (light shaded) and using $\what{\hyp}$ \cref{eq:marginal_likelihood} (dark shaded).}
    \label{fig:example_1_eb}
\end{figure}

\begin{figure}
    \centering
    \setlength{\figurewidth}{0.55\columnwidth}
    \setlength{\figureheight}{0.35\columnwidth}
	\ifarxiv
		\includegraphics{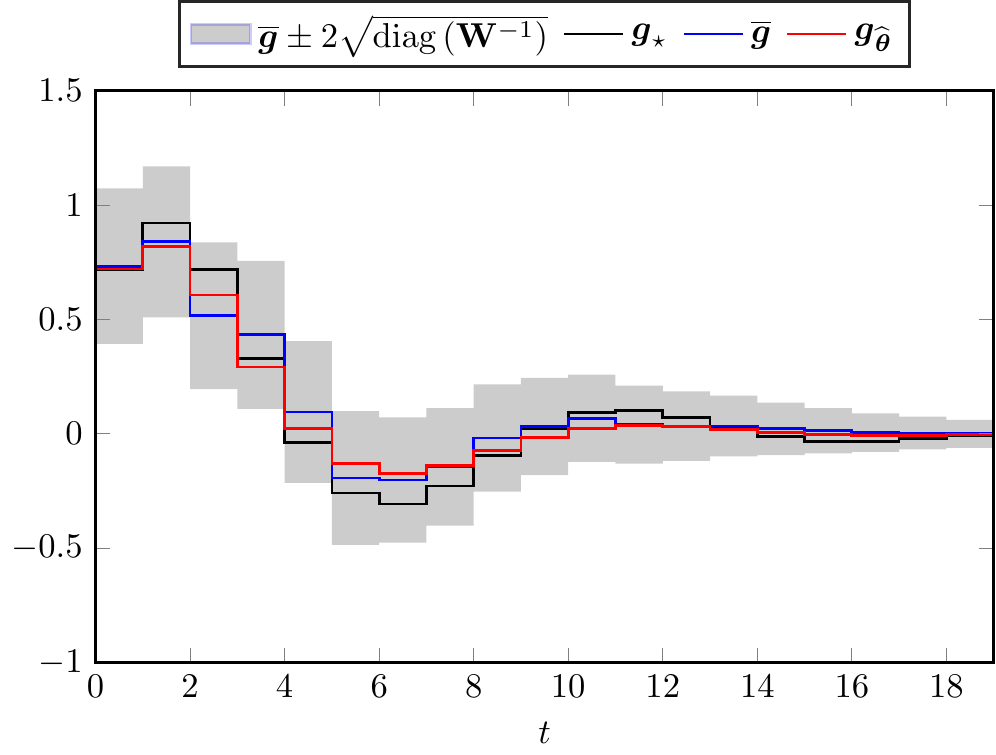}
	\else
		\tikzsetnextfilename{example_1_krm}
		\input{fig/example_1_krm.pgf}
	\fi

    \caption{Same scenario as in \cref{fig:example_1_prior} but with a covariance matrix $\kernelK(\what{\hyp})$ that is automatically tuned using \cref{eq:marginal_likelihood}.}
    \label{fig:example_1_krm}
\end{figure}

\end{exmp}
\setcounter{exmpcnt}{\mytmp}

\section{Numerical illustrations}
In this section, we evaluate the proposed decision-theoretic approach, which we call \gls{brm} when incorporating prior information in the form of a distribution on the impulse response $\imp$.
For the prior covariance matrix $\kernelK$, we consider a \gls{dc} covariance function \cref{eq:dckernel} and tune the hyperparameters using \cref{eq:marginal_likelihood}, cf. \cref{sec:modeling_prior_information}.
The method is summarized in \cref{alg:method}. The identified models are compared with those of \gls{pem}.

\begin{algorithm}
 \caption{Method for model class $\modellinear_{\para}$} \label{alg:method}
\begin{algorithmic}[1]
    \State Input: $\dataset = \{ (u(1),y(1)), \: \dots \:  (u(N), y(N)) \}$
    \State Form matrix $\H$.
    \State Tune hyperparameters $\hyp$.
    \State Form $\impmean$ and $\weight$ using \cref{eq:prior_mu_and_weight}.
    \State Solve problem in \cref{eq:optimaldecision}.
    \State Output: $\paradec$
\end{algorithmic}
\end{algorithm}

We consider the following data generating system
\begin{equation*}
	G_\star(q) = \frac{0.41}{1 - 1.82q^{-1} + 2.04q^{-2} - 1.27q^{-3} + 0.46 q^{-4}}.
\end{equation*}
This fourth-order system has two resonance peaks in the system frequency response.
Similar to \cref{ex:2}, we generate the input as a white process with variance $1$ and consider disturbances with variance 2.
For each experiment we obtain a dataset $\dataset$ of length $N$ and evaluate the decision by the normalized squared error
\begin{equation*}
	\log \frac{\| \imp_\star - \imp_{\paradec} \|^2}{\| \imp_\star \|^2}
\end{equation*}
in logarithmic units for clear comparisons.
We repeat the experiments 100 times and study the distribution of errors.

In the first evaluation, we vary the size of the dataset using small, medium and large $N$.
The model class $\modellinear_{\para}$ under consideration matches the form of the unknown system.
The results are show in \cref{fig:sim_N}.
Note that for small sample sizes, $N=30$, \gls{pem}-based decisions frequently produces errors greater 0, which is higher than simply assuming an impulse response with zero coefficients.
The effect of regularization in \gls{brm} is significant and the performance by \gls{pem} is only matched by doubling the sample size to $N=60$.
When $N=120$, the median errors of both methods are comparable, but \gls{brm} has a narrower dispersion than \gls{pem}.  
\begin{figure}
    \centering
    \setlength{\figurewidth}{0.85\columnwidth}
    \setlength{\figureheight}{0.5\columnwidth}
	\ifarxiv
		\includegraphics{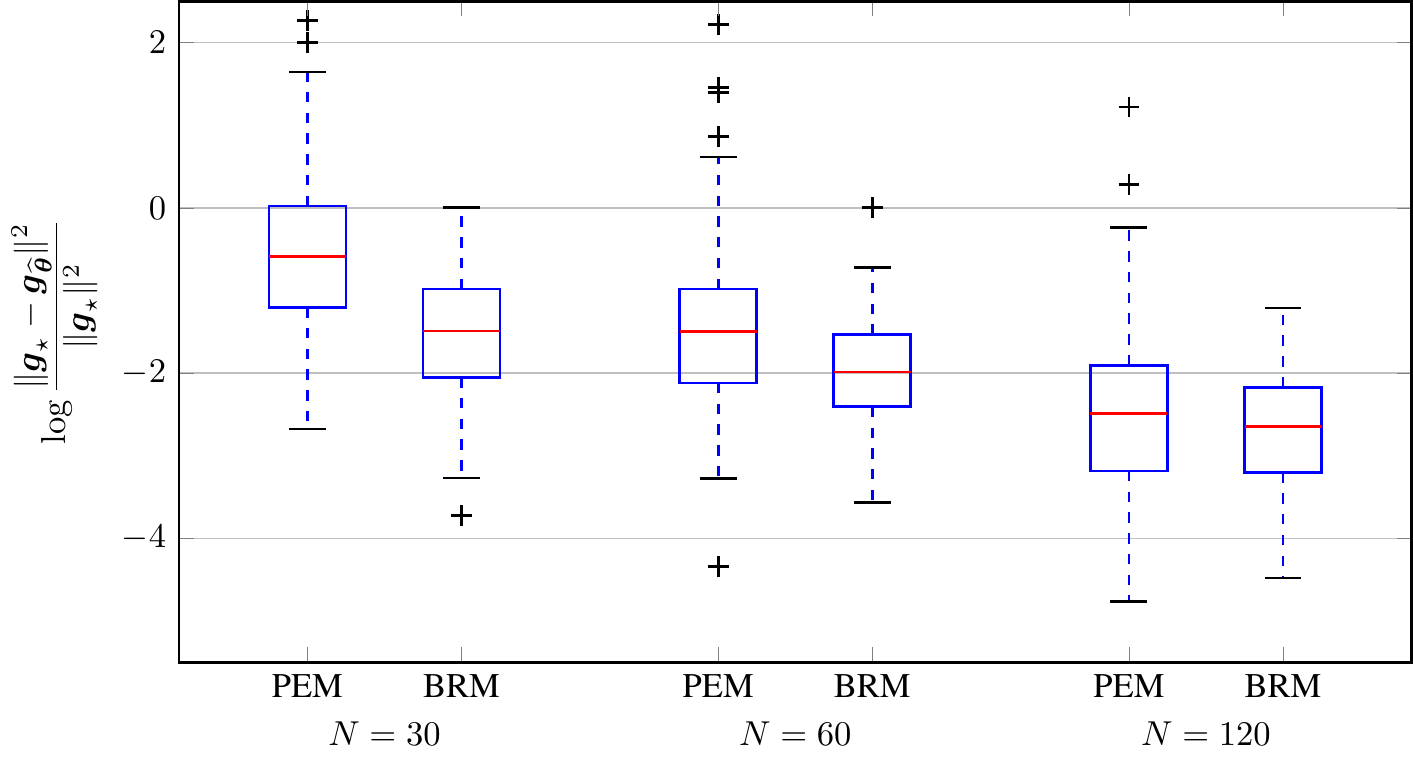}
	\else
		\tikzsetnextfilename{boxplot_N}
		\input{fig/boxplot_N.pgf}
	\fi

    \caption{Box-plots of errors for \gls{pem} and \gls{brm}. Numbers indicate the sample size $N$.}
    \label{fig:sim_N}
\end{figure}

In the second evaluation, we vary the order $n_f$ of the model class $\modellinear_{\para}$ when $N=60$.
Specifically, we consider $n_f$ to be lower, equal to, and greater than the unknown system model order.
The results are illustrated in \cref{fig:sim_nf}, where it is seen that for $n_f=2$, \gls{pem} and \gls{brm} perform similarly.
As the overparameterization increases, \gls{pem} overfits and produces many decisions with normalized errors above 0.
By contrast, the robustness of \gls{brm} to overparameterization is notable.

\begin{figure}
    \centering
    \setlength{\figurewidth}{0.80\columnwidth}
    \setlength{\figureheight}{0.5\columnwidth}
	\ifarxiv
		\includegraphics{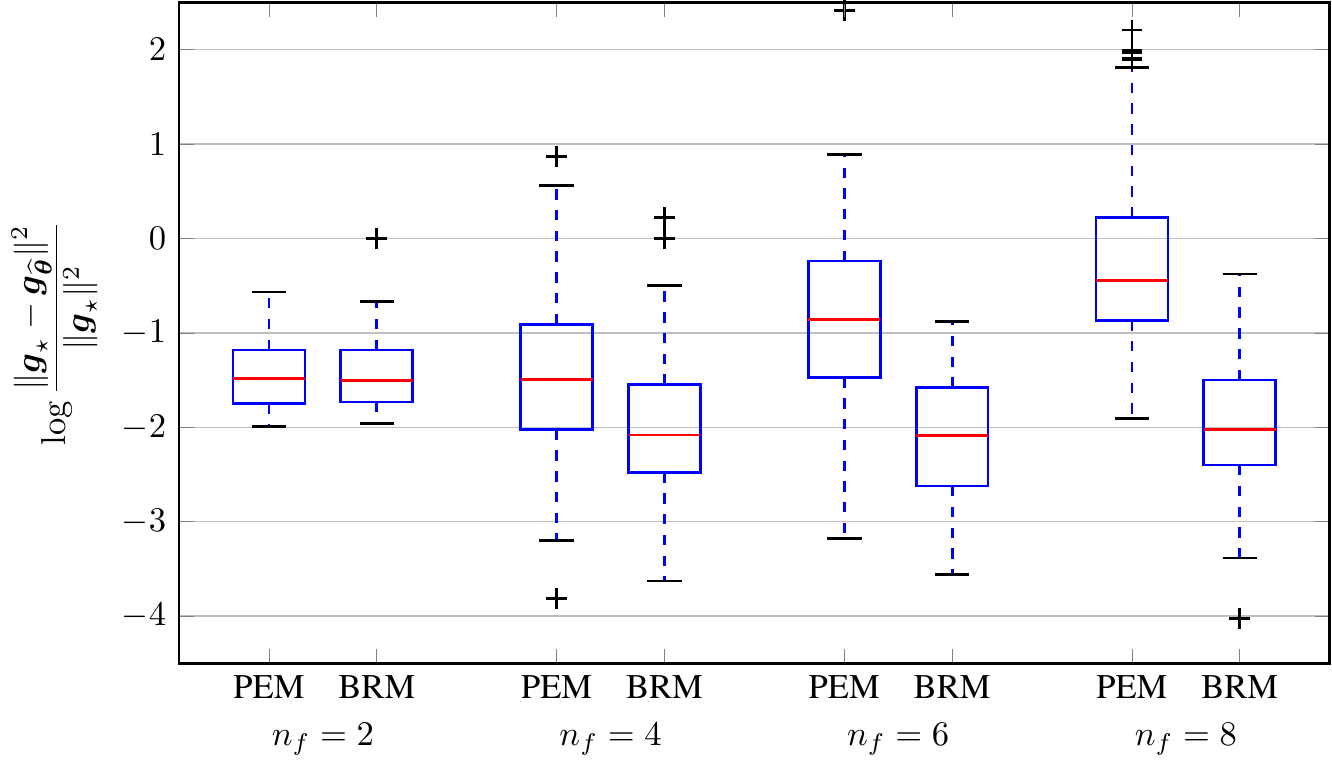}
	\else
		\tikzsetnextfilename{boxplot_nf}
		\input{fig/boxplot_nf.pgf}
	\fi

    \caption{Box-plots of errors for \gls{pem} and \gls{brm}. Numbers indicate the model order $n_f$. Here $N=60$ samples are used. The order of the data generating system is 4.}
    \label{fig:sim_nf}
\end{figure}

\section{Conclusions}

We developed a decision-theoretic approach that can be used to identify parametric linear time-invariant models. This approach bridges the gap between the existing classical and regularized identification approaches. We showed that it leads to a parametric identification method with a simple data-adaptive regularization by assigning a prior distribution over the system impulse response. This incorporates prior information of the system and dispels both the need to assume a prior distribution over the model parameters as well as the assumption that the unknown system belongs to the model class. The numerical experiments illustrated that the regularization yields a method that is robust to both small sample-sizes and overparameterization.

In future work, we will study how the approach can be extended to a broader model class. We also consider the equivalent formulation of the risk function in the frequency domain.

\appendix
\section{Appendix}
\subsection{Risk expression}\label{app:risk}
Let us start by noting that 
\begin{align*}
	\| \imp - \imp_{\para} \|^2_{\weight}
	    = \| \imp - \impmean + \impmean - \imp_{\para} \|^2_{\weight}
	    = \| \imp - \impmean\|^2_{\weight} + \| \impmean - \imp_{\para}\|^2_{\weight} + 2(\imp - \impmean)^{\T}(\impmean - \imp_{\para}),
\end{align*}
for any vector $\impmean$ of appropriate size. 
Now, if we choose $\impmean = \CExp{\imp \mid \dataset}{\imp}$, we obtain
\begin{align*}
	\CExp{\imp \mid \dataset}{ \| \imp - \imp_{\para}\|^2_{\weight}} 
		&= \E_{\imp \mid \dataset} \Bigg[\| \imp - \impmean\|^2_{\weight} + \| \impmean - \imp_{\para}\|^2_{\weight} + 2(\imp - \impmean)^{\T}(\impmean - \imp_{\para})\Bigg] 
		= \trace{\weight \mat{\Sigma}} + \|\impmean - \imp_{\para} \|^{2}_{\weight}.
\end{align*}

\subsection{PEM as a decision}\label{app:pem_as_decision}
The objective function used in the prediction error method can be written as
\begin{equation}
	\what{\para} = \argmin_{\para} \frac{1}{2} \big\| \data - \modeldata_{\para} \big\|^{2},
\end{equation}
where $\modeldata_{\para}$ is the predicted  output of the parametric system given the input.
The predicted output is given by
\begin{equation}
    \modeldata_{\para} = \H \imp_{\para}.
\end{equation}
Let $\impmean = \H^{\pinv}\data$, where $\H^{\pinv}$ is a pseudoinverse of $\H$.
We get
\begin{align*}
	\what{\para}
		&= \argmin_{\para} \frac{1}{2} \| \data - \modeldata_{\para}\|^2 \\
		&= \argmin_{\para} \frac{1}{2} \| \data - \H\impmean + \H\impmean - \modeldata_{\para}\|^2 \\
		&= \argmin_{\para} \frac{1}{2} \Big[ \| \data - \H\impmean\|^2 + 2(\data - \H\impmean)^{\T}\H(\impmean - \imp_{\para}) + \| \H(\impmean - \imp_{\para})\|^2 \Big] \\
		&\Bigg/ \impmean = \H^{\pinv}\data  \implies (\data -\H \impmean)^{\T}\H = (\data -\H \H^{\pinv}\data)^{\T}\H = \0 \Bigg/ \\
		&= \argmin_{\para} \frac{1}{2} \| \impmean - \imp_{\para}\|^{2}_{\H^{\T}\H}.
\end{align*}
Assuming that the generating system is an \gls{oe} system, that is $\data = \H\imp + \noise$, we have that \cite{Soderstrom&Stoica1988_system}
\begin{equation*}
    \Exp{\impmean} = \H^{\pinv}\H\imp \quad\text{and}\quad \Cov(\impmean) = \lambda(\H^{\T}\H)^{\pinv}.
\end{equation*}
Note that $\lambda^{-1}\H^{\T}\H$ is a precision matrix and therefore a straight-forward choice for $\weight$. Thus the risk
\begin{equation}
	\risk(\para) = \frac{1}{2} \| \impmean - \imp_{\para}\|^2_{ \lambda^{-1}\H^{\T}\H}.
\end{equation}
leads to PEM since the risk-minimizing decision is invariant with respect to $\lambda$.


\section*{Acknowledgement}
The authors would like to thank Prof. Petre Stoica and Prof. Lennart Ljung for valuable comments.

\bibliographystyle{plain}
\bibliography{ref_decision}

\begin{thebibliography}{10}

\bibitem{Berger1985_statistical}
J.O. Berger.
\newblock {\em Statistical decision theory and {B}ayesian analysis}.
\newblock Springer Series in Statistics. Springer, 1985.

\bibitem{Chernoff&Moses1986_elementary}
Herman Chernoff and Lincoln~E Moses.
\newblock {\em Elementary decision theory}.
\newblock Dover, 1986.

\bibitem{Ljung1998_sysid}
L.~Ljung.
\newblock {\em System Identification: Theory for the User}.
\newblock Pearson Education, 1998.

\bibitem{Murphy2012machine}
Kevin~P Murphy.
\newblock {\em Machine learning: a probabilistic perspective}.
\newblock MIT press, 2012.

\bibitem{Ninness&Henriksen2010_bayesian}
Brett Ninness and Soren Henriksen.
\newblock Bayesian system identification via {M}arkov chain {M}onte {C}arlo
  techniques.
\newblock {\em Automatica}, 46(1):40--51, 2010.

\bibitem{Peterka1981bayesian}
V{\'a}clav Peterka.
\newblock Bayesian system identification.
\newblock {\em Automatica}, 17(1):41--53, 1981.

\bibitem{Pillonetto&Giuseppe2010_kernel}
Gianluigi Pillonetto and Giuseppe De~Nicolao.
\newblock A new kernel-based approach for linear system identification.
\newblock {\em Automatica}, 46(1):81--93, 2010.

\bibitem{PillonettoEtAl2014_kernelsurvey}
Gianluigi Pillonetto, Francesco Dinuzzo, Tianshi Chen, Giuseppe De~Nicolao, and
  Lennart Ljung.
\newblock Kernel methods in system identification, machine learning and
  function estimation: A survey.
\newblock {\em Automatica}, 50(3):657--682, 2014.

\bibitem{Rasmussen&Williams2006_gaussian}
Carl~Edward Rasmussen and Christopher~KI Williams.
\newblock {\em Gaussian processes for machine learning}, volume~1.
\newblock MIT press Cambridge, 2006.

\bibitem{Robert2007_bayesian}
Christian Robert.
\newblock {\em The Bayesian choice: from decision-theoretic foundations to
  computational implementation}.
\newblock Springer Science \& Business Media, 2007.

\bibitem{Soderstrom&Stoica1988_system}
T.~S\"oderstr\"om and P.~Stoica.
\newblock {\em System identification}.
\newblock Prentice-Hall, Inc., 1988.

\end{thebibliography}

\end{document}